\documentclass{article}
\usepackage{graphicx} 
\usepackage{xcolor}
\usepackage{subcaption}
\usepackage{hyperref}
\usepackage[margin=1in]{geometry}
\usepackage[fontsize=11pt]{fontsize}
\usepackage{amsmath,amsfonts,amssymb,amsthm}
\usepackage{mathtools}
\newtheorem{lemma}{Lemma}
\newtheorem{theorem}{Theorem}
\newtheorem{fact}{Fact}

\newtheorem{observation}{Observation}

\newtheorem{definition}{Definition}

\usepackage{mathtools}

\title{Critical Thresholds for Maximum Cardinality Matching on General Hypergraphs}
\author{Christopher Sumnicht$\dagger$, Jamison W. Weber$\dagger$,\\ Dhanush R. Giriyan, Arunabha Sen\\ 
School of Augmented Intelligence, Arizona State University, Tempe, AZ\\
{\small $\dagger$Contributed equally.}}
\date{July 2024}

\begin{document}

\maketitle

\begin{abstract}
    Significant work has been done on computing the ``average'' optimal solution value for various $\mathsf{NP}$-complete problems using the Erd\"{o}s-R\'{e}nyi model to establish \emph{critical thresholds}. Critical thresholds define narrow bounds for the optimal solution of a problem instance such that the probability that the solution value lies outside these bounds vanishes as the instance size approaches infinity.
    In this paper, we extend the Erd\"{o}s-R\'{e}nyi model to general hypergraphs on $n$ vertices and $M$ hyperedges.
    We consider the problem of determining critical thresholds for the largest cardinality matching, and we show that for $M=o(1.155^n)$ the size of the maximum cardinality matching is almost surely 1. On the other hand, if $M=\Theta(2^n)$ then the size of the maximum cardinality matching is $\Omega(n^{\frac12-\delta})$ for an arbitrary $\delta >0$.
    Lastly, we address the gap where $\Omega(1.155^n)=M=o(2^n)$ empirically through computer simulations. 
\end{abstract}

\section{Introduction}

The seminal work of \cite{erd6s1960evolution} has established that many graph properties have sharp concentrations over random distributions of graphs. That is, given a random graph with a certain fraction of edges, there is a \emph{critical threshold} where slightly above (or sometimes below) this threshold the probability of a graph satisfying such a property approaches $1$ as the number of vertices increases, and vice versa. This work has given extensive insights to various $\mathsf{NP}$-complete problems to graphs such as \textsf{CLIQUE} (see \cite{bollobas2013modern}), \textsf{INDEPENDENT-SET} (see \cite{frieze2016introduction}), and \textsf{DOMINATING-SET} (see \cite{wieland2001domination} and \cite{harutyunyan2015total}). A large amount of research has been done on critical thresholds for random objects beyond graphs such as the \textsf{SAT} problem (see for instance, \cite{coja2014asymptotic} and \cite{diaz2009satisfiability}). Our work extends this by considering  matchings on random hypergraphs.
Note that maximum cardinality matching in hypergraphs is \textsf{NP}-hard as it generalizes \textsf{3-DIMENSIONAL-MATCHING}~\cite{Karp1972}.

Before giving the problem statement in full, we first present some definitions.
First we consider the $\mathcal{G}(n,m)$ version of the Erd\"{o}s-R\'{e}nyi model~\cite{newman2018networks}.
That is, $\mathcal{G}(n,m)$ represents the uniform distribution over \emph{all possible} undirected graphs with \emph{exactly} $n$ vertices and $m$ edges. 
We extend this model to \emph{all possible} undirected hypergraphs, denoted as $\mathcal{H}(n,M)$.
A sample $(\mathcal{U},\mathcal{S})\sim\mathcal{H}(n,M)$ is given by a vertex set $\mathcal{U}$ with $|\mathcal{U}|=n$ and a hyperedge set $\mathcal{S}$ with $|\mathcal{S}|=M$ where $\mathcal{S}\subseteq\mathcal{P}(\mathcal{U})\setminus\emptyset$.
As such, an edge in $\mathcal{S}$ has at least one and at most $n$ endpoints. 
Denote $\mathcal{S}_k\subseteq\mathcal{S}$ with $|\mathcal{S}_k|=k\le M$.

\begin{definition}\label{def:hypermatching}
A \textbf{hyper-matching} is a collection $\mathcal{S}_k\subseteq\mathcal{S}$ where for all distinct $S_i,S_j\in \mathcal{S}_k$ $$S_i\cap S_j=\emptyset.$$
The \textbf{hyper-matching number} is the largest $k$  such that there exists an $\mathcal{S}_k$ that is a hyper-matching.
\end{definition}

Let $X_k$ be the random variable that denotes the number of hyper-matchings of size $k$ that can be formed from $\mathcal{S}$.
Then we denote the expected value of $X_k$ as $\mathbb{E}[X_k]$ and its variance as $\mathsf{Var}[X_k]$.
For random variables $X,Y$, denote $\mathsf{Cov}[X,Y]$ as the pairwise covariance of $X$ and $Y$.
We ask whether there exist functions $f(n,M)$ (lower bound) and $g(n,M)$ (upper bound) where $g(n,M)\geq f(n,M)$ with the following properties:
    \begin{enumerate}
    \item {$\lim_{n \to \infty} \Pr[X_{f(n,M)} \geq 1] = 1$}.
    \item{$\lim_{n\to\infty}\Pr[X_{g(n,M)} \geq 1] = 0$}.
    \end{enumerate}
Ideally, the gap $g(n,M)-f(n,M)$ should be as small as possible.
As such, functions $f(n,M)$ and $g(n,M)$ characterize a narrow range of values wherein the hyper-matching number lies for almost all hypergraphs.
If functions $g$ and $f$ satisfy this property, then we say that the hyper-matching number is \emph{almost surely} in $[f(n,M),g(n,M)]$.
We present the following theorem as our primary contribution.
\begin{theorem}\label{theroem:main}
Let $(\mathcal{U},\mathcal{S})\sim\mathcal{H}(n,M)$.
If $M=o(1.155^n)$,
then the hyper-matching number of $(\mathcal{U},\mathcal{S})$ is almost surely 1.
If $M=\Theta(2^n)$, then the hyper-matching number of $(\mathcal{U},\mathcal{S})$ is almost surely in $[\Omega(n^{\frac12-\delta}),n]$ for arbitrary $\delta>0$.
\end{theorem}
\noindent Notice that Theorem~\ref{theroem:main} omits a gap where $\Omega(1.155^n)=M=o(2^n)$.
We refer to this as the \emph{behavioral gap} and do not address this range analytically, but rather empirically. 
It is in this gap that we observe a phase shift where the hyper-matching number tends from unity to sublinear in $n$. 

In Section~\ref{sec:relatedwork} we discuss related work. In Section~\ref{sec:preliminaries} we detail preliminaries. In Section~\ref{sec:analysis} we present the proof of Theorem~\ref{theroem:main}.
In Section~\ref{sec:experimentalvalidation} we present two empirical studies using computer simulations---namely,
we validate our theory in the case where $M$ is small in Section~\ref{sec:smallM}, and we analyze the hyper-matching number empirically in the case where $M$ is in the behavioral gap in Section~\ref{sec:behavioralexp}.
Lastly, in Section~\ref{sec:discussion} we provide a discussion of our results.

\section{Related Work}
\label{sec:relatedwork}

Certainly the study of random hypergraphs is not new. In fact, there are many competing definitions of what constitutes a ``random'' hypergraph (see for instance \cite{barthelemy2022class} for various different random hypergraph models). In the work of~\cite{ghoshal2009random}, for instance, the authors study applications of random tripartite graphs (a generalization of bipartite graphs to three partitions) in order to better understand the structure of social media sites such as Flickr. Moreover, even the study of matchings in hypergraphs is not new with many applications in ``capital budgeting, crew scheduling, facility location, scheduling airline flights, forming a coalition structure in multi-agent systems, and determining the winners in combinatorial auctions," as stated in~\cite{hanguir2021distributed}.

There has been a significant amount of research focused on matchings in \textit{uniform} hypergraphs (such as \cite{alon2005hypergraph}, \cite{kim2003perfect}, and \cite{frieze1995perfect}) as opposed to hypergraphs in general. More importantly, to the best of our knowledge, there is no notion of computation for critical thresholds of maximum cardinality matchings in random hypergraphs in the literature. Instead, significant work has been placed on finding sufficient criterion for perfect (or nearly perfect) matchings (see for instance \cite{alon2005hypergraph} and \cite{keevash2018hypergraph}), and work has been done for finding matchings (in particular, perfect matchings or nearly perfect matchings) in hypergraphs (see for instance \cite{harris2019distributed} and \cite{hanguir2021distributed}).

\section{Preliminaries}\label{sec:preliminaries}
In this section we recall existing results that we apply to our analysis.
We apply extensively the binomial coefficient, denoted as $\binom{n}{k}$, which represents the number of ways to select $k$ objects from $n$ objects.
It is known that
\begin{fact}
(Bounds on Binomial Coefficient)
\begin{equation}\label{eq:binom}
\left(\frac{n}{k}\right)^k\le \binom{n}{k} \le \left(\frac{ne}{k}\right)^k, 0\le k\le n.
\end{equation}
\end{fact}
Hence, the error is at most proportional to $e^k$.
Next, we will refer to the Stirling number of the second kind~\cite{graham1994concrete}, denoted as $\left\{ {n \atop k} \right\}$, which represents the number of ways to partition a set of $n$ objects into $k$ non-empty subsets.
This number is given by
\begin{fact}\label{fact:stirling}
(Stirling Number of the Second Kind)
\begin{equation}
    \left\{ {n \atop k} \right\}=\sum_{i=0}^k\frac{(-1)^{k-i}i^n}{(k-i)!i!},
    \end{equation}
\end{fact}
Moreover, it is known that
\begin{fact}\label{fact:stirlingbound}
(Bounds on Stirling Number of the Second Kind)
\begin{equation}
    \frac12(k^2+k+2)k^{n-k-1}-1 \le \left\{ {n \atop k} \right\} \le \frac12\binom{n}{k}k^{n-k},\forall n\geq 2,\forall 1\le k \le n-1.
    \end{equation}
\end{fact}
It is also known that
\begin{fact}\label{fact:binomial}
    (Binomial Theorem)
    $\sum_{j=0}^{n} {n \choose j}x^jy^{n-j} = (x + y)^n$ and in the special case when $x=y= 1$, we have $\sum_{j=0}^{n} {n \choose j} = 2^n$
\end{fact}

We also apply the following concentration inequalities, given by the following facts.
\begin{fact}\label{fact:markov}
    (Markov Inequality) Let $X$ be a non-negative random variable and let $a>0$. Then
    \begin{equation}
        \Pr[X\geq a]\le \frac{\mathbb{E}[X]}{a}
    \end{equation}
\end{fact}
\begin{fact}\label{fact:chebyshev}
(Chebyshev Inequality)
Let $X$ be a random variable with finite variance $\textsf{Var}[X]$.
Then
\begin{equation}
    \Pr[|X-\mathbb{E}[X]|\geq \mathbb{E}[X]]\le\frac{\textsf{Var}[X]}{(\mathbb{E}[X])^2}
\end{equation}
Lastly, we apply the following expansion.
\end{fact}
\begin{fact}\label{fact:exptaylor}
    (Taylor Expansion of $e^x$ about $x=0$)
    $$e^x = \sum_{j=0}^{\infty}\frac{x^j}{j!},\forall x$$
\end{fact}
Note that this series converges absolutely for all $x \in \mathbb{R}$. 
\section{Analysis}
\label{sec:analysis}
For the remainder of this section we prove Theorem~\ref{theroem:main}.
The following lemma quantifies the probability that $\mathcal{S}_k\subseteq\mathcal{P}(\mathcal{U})\setminus \emptyset$ is a hypermatching. Let $\mathsf{Match}(\mathcal{S}_k)$ be the event that $\mathcal{S}_k$ is a hypermatching.
Note that $\mathcal{S}_k$ is a hypermatching if and only if all sets in $\mathcal{S}_k$ are pairwise disjoint (see Definition~\ref{def:hypermatching}) and $\mathcal{S}_k\subseteq \mathcal{S}$.
Let $\mathsf{PWD}(\mathcal{S}_k)$ denote the event that all sets in $\mathcal{S}_k$ are pairwise disjoint.

\begin{lemma}\label{lemma:probability}
    Let $(\mathcal{U},\mathcal{S})\sim\mathcal{H}(n,M)$ and let $n_\mathcal{S}=|\bigcup_{S_i\in\mathcal{S}}S_i|$.
    Now let $\mathcal{S}_k$ be an arbitrary set from $\binom{\mathcal{P}(\mathcal{U})\setminus\emptyset}{k}$.
    Then we have
    $$\Pr[\mathsf{Match}(\mathcal{S}_k) ]=\frac{\sum_{j=1}^{n_{\mathcal{S}}}\binom{n_{\mathcal{S}}}{j}\left\{ {j \atop k} \right\}}{\binom{2^n-1}{k}}.$$
\end{lemma}
\begin{proof}
    The probability that $\mathcal{S}_k$ is a matching is equal to the joint probability that all sets in $\mathcal{S}_k$ are pairwise disjoint and $\mathcal{S}_k\subseteq\mathcal{S}$.
    First suppose $\mathcal{S}_k\subseteq\mathcal{S}$.
    Since any $(\mathcal{U},\mathcal{S})$ is equiprobable, determining the probability that $\mathcal{S}_k$ is a matching given that $\mathcal{S}_k\subseteq\mathcal{S}$ reduces to counting the number of matchings of size $k$ that can be formed from the vertices incident to edges in $\mathcal{S}$.
    The Stirling number $\left\{ {n_{\mathcal{S}} \atop k} \right\}$ of the second kind represents the number of ways to partition a set of $n_\mathcal{S}$ elements into $k$ non-empty subsets by Fact~\ref{fact:stirling}.
    A matching, however, is not necessarily a partition of the vertex set induced by $\mathcal{S}$. 
    Instead, it may consist of fewer incident vertices than $n_\mathcal{S}$.
    Hence, the number of matchings of size $k$ is given by $\sum_{j=1}^{n_\mathcal{S}}\binom{n_\mathcal{S}}{j}\left\{ {j \atop k} \right\}$. 
    That is, the number of non-empty $k$-sized partitions over all subsets of vertices incident to edges in $\mathcal{S}$.
    We must also compute the probability of the event $\mathcal{S}_k\subseteq\mathcal{S}$ for our calculation.
    Of the $\binom{2^n-1}{k}$ ways to select $\mathcal{S}_k$, a total of $\binom{M}{k}$ of these are subsets of $\mathcal{S}$. 
     Again since all $(\mathcal{U},\mathcal{S})$ are equiprobable, we compute the bound on our desired probability as
    $$
      \Pr[\mathsf{Match}(\mathcal{S}_k)]=
    \Pr[\mathsf{PWD}(\mathcal{S}_k) \land \mathcal{S}_k\subseteq \mathcal{S}]=\Pr[\mathsf{PWD}(\mathcal{S}_k)\mid \mathcal{S}_k\subseteq\mathcal{S}]\Pr[\mathcal{S}_k\subseteq\mathcal{S}]=
    $$$$
    \frac{\sum_{j=1}^{n_{\mathcal{S}}}\binom{n_\mathcal{S}}{j}\left\{ {j \atop k} \right\}}{\binom{M}{k}}\cdot \frac{\binom{M}{k}}{\binom{2^n-1}{k}}=\frac{\sum_{j=1}^{n_{\mathcal{S}}}\binom{n_\mathcal{S}}{j}\left\{ {j \atop k} \right\}}{\binom{2^n-1}{k}}
    $$
\end{proof}

We now quantify bounds on the expected number of hyper-matchings of a certain size.
Recall that $X_k$ is the number of hyper-matchings of size $k$ that can be formed from $\mathcal{S}$.
The following lemma provides bounds on the expected value of this random variable for a given $k$.

\begin{lemma}\label{lemma:expectation}
    $$\frac{M}{2^n} {M \choose k}\frac{\sum_{j=1}^n{n \choose j}\left\{ {j \atop k} \right\}}{\binom{2^n-1}{k}}\le\mathbb{E}[X_k]\le {M \choose k}\frac{\sum_{j=1}^n{n \choose j}\left\{ {j \atop k} \right\}}{\binom{2^n-1}{k}}$$
\end{lemma}
\begin{proof}
Variable $X_k$ can be decomposed into a sum $X_k=\sum_{\mathcal{S}_k\subseteq\mathcal{P}(\mathcal{U})\setminus\emptyset}Y_{\mathcal{S}_k}$, where 
$$Y_{\mathcal{S}_k}=\begin{cases}
    1 &\text{ if }\mathcal{S}_k\text{ is a matching}\\
    0 &\text{ otherwise}
\end{cases}.$$
Again let $n_\mathcal{S}=|\bigcup_{S_i\in\mathcal{S}}S_i|$.
Then by linearity of expectation and Lemma~\ref{lemma:probability}, the expected number of matchings of size $k$ is given by
$$
    \mathbb{E}[X_k]= \mathbb{E}[\sum_{\mathcal{S}_k\subseteq\mathcal{S}}Y_{\mathcal{S}_k}]= \sum_{\mathcal{S}_k\subseteq\mathcal{S}}\mathbb{E}[Y_{\mathcal{S}_k}]= \sum_{\mathcal{S}_k\subseteq\mathcal{S}}\Pr[\textsf{Match}(\mathcal{S}_k)]=\binom{M}{k}\Pr[\textsf{Match}(\mathcal{S}_k)]$$$$= {M \choose k}\frac{\sum_{j=1}^{n_{\mathcal{S}}}{n_{\mathcal{S}} \choose j}\left\{ {j \atop k} \right\}}{\binom{2^n-1}{k}}
$$
As $n_{\mathcal{S}}\le n$ clearly holds, we obtain our upper bound of ${M \choose k}\frac{\sum_{j=1}^{n}{n \choose j}\left\{ {j \atop k} \right\}}{\binom{2^n-1}{k}}$.

For the lower bound, we first derive a lower bound on $\Pr[n=n_{\mathcal{S}}]$.
Consider that one of the many possible ways to satisfy $n=n_{\mathcal{S}}$ is if $\mathcal{U}\in\mathcal{S}$.
Due to the uniform selection of $(\mathcal{U},\mathcal{S})$, the probability that a fixed edge is in $\mathcal{S}$ where $|\mathcal{S}|=M$ is given by $\frac{M}{2^n-1}\geq\frac{M}{2^n}$.
Hence, this probability also represents a lower bound on $\Pr[n=n_\mathcal{S}]$.
By the law of total expectation, we obtain the desired lower bound on the expectation as
\begin{align*}
    &\mathbb{E}[X_k]\\
    &=\mathbb{E}[X_k\mid n=n_{\mathcal{S}}]\Pr[n=n_{\mathcal{S}}]+\mathbb{E}[X_k\mid n> n_{\mathcal{S}}]\Pr[n> n_{\mathcal{S}}]\\
    &\geq \mathbb{E}[X_k\mid n=n_{\mathcal{S}}]\Pr[n=n_{\mathcal{S}}]\\
    &\geq \mathbb{E}[X_k\mid n=n_{\mathcal{S}}]\frac{M}{2^n}\\
    &= \frac{M}{2^n}{M \choose k}\frac{\sum_{j=1}^{n}{n \choose j}\left\{ {j \atop k} \right\}}{\binom{2^n-1}{k}}\\
\end{align*}
\end{proof}

We can use these results to quickly handle the cases where $M=o(1.155^n)$.
We show that in this case, the hyper-matching number (surprisingly) necessarily converges to unity, i.e. $g(n,M)=1$.

\begin{lemma}\label{lemma:smallM}
Suppose $M= o(1.155^n)$.
Then $$\lim_{n\to\infty}\Pr[\text{There exists a matching of size at least }2]=0.$$
\end{lemma}
\begin{proof}
    It suffices to show that the probability that there exists a hyper-matching of size exactly 2 tends to zero in the limit since the probability that there exists a hyper-matching of size $k>2$ given there is no hyper-matching of size exactly 2 is precisely zero.
    That is, if there is no hyper-matching of size 2, then there is necessarily no hyper-matching of size larger than 2.
    This holds since any subset of a hyper-matching is also a hyper-matching.
    By Lemma~\ref{lemma:expectation} and Facts~\ref{fact:binomial},\ref{fact:markov}, we have
    \begin{align*}
        &\Pr[X_2\geq1] \le \mathbb{E}[X_2]\\
        &\le\binom{M}{2}\frac{\sum_{j=1}^n{n \choose j}\left\{ {j \atop 2} \right\}}{\binom{2^n-1}{2}}\\
        &=\frac{M(M-1)}{2}\frac{\sum_{j=1}^n{n \choose j}\left\{ {j \atop 2} \right\}}{\frac{(2^n-1)(2^n-2)}{2}}\\
        &=M(M-1)\frac{\sum_{j=1}^n{n \choose j}\sum_{i=0}^2\frac{(-1)^{2-i}i^j}{(2-i)!i!}}{(2^n-1)(2^n-2)}\\
        &=M(M-1)\frac{\sum_{j=1}^n{n \choose j}\left(\frac{(-1)^{1}1^j}{(2-1)!1!}+\frac{(-1)^{0}2^j}{(2-2)!2!}\right)}{(2^n-1)(2^n-2)}\\
        &=M(M-1)\frac{\sum_{j=1}^n{n \choose j}\left(-1+\frac{2^j}{(0)!2!}\right)}{(2^n-1)(2^n-2)}\\
        &=M(M-1)\frac{\sum_{j=1}^n{n \choose j}\left(2^{j-1}-1\right)}{(2^n-1)(2^n-2)}\\
        &\le M(M-1)\frac{\sum_{j=1}^n{n \choose j}\left(2^{j}\right)}{(2^n-1)(2^n-2)}\\
        &\le M(M-1)\frac{3^n}{(2^n-1)(2^n-2)}\\
        &= M(M-1)\frac{3^n}{\Theta(2^{2n})}\\
        &= M(M-1)\frac{e^{n\ln3}}{\Theta(e^{2n\ln2})}\\
        &= M(M-1)\Theta(e^{n(\ln3-2\ln{2})})\\
    \end{align*}
    This expression tends to zero as $n\to\infty$ since $2\ln{2}>\ln3$.
    This holds for $M=o(e^{-\frac{n}{2}(\ln3-2\ln2)})\approx o(e^{0.144n})=o(1.155^n)$ and taking $M$ smaller only shrinks this expression.
\end{proof}
 
Now we assume $M=\Theta(2^n)$.
For an upper bound on the hyper-matching number, we make the following observation.
\begin{observation}\label{obs:upperbound}
    If $M=2^n-1$, then the hyper-matching number is $n$ with probability 1. 
\end{observation}
\noindent This holds since in this case the hypergraph is \emph{complete}, meaning all possible edges are present,
and the size of the largest partition of $n$ objects into non-empty sets is $n$. 
This implies a tight upper bound of $g(n,M)=n$, and so now we focus on the lower bounds.

\subsection{Lower Bounds}
For the function $f:=f(n,M)$, it suffices to show that the probability that $X_f=0$ tends to zero as $n$ approaches infinity.
By Fact~\ref{fact:chebyshev} and since $X_f$ is non-negative, we can bound this probability as follows.
$$\Pr[X_f=0]\le\Pr[X_f\le0\lor X_f\geq 2\mathbb{E}[X_f]]=\Pr[|X_f-\mathbb{E}[X_f]|\geq\mathbb{E}[X_f]]\le\frac{\mathsf{Var}[X_f]}{(\mathbb{E}[X_f])^2}$$
Hence, we must show that the right-most expression tends to zero in the limit, which requires bounding $\mathsf{Var}[X_f]$ from above.
Since some $\mathcal{S}_k$ may overlap, we do not necessarily have independence, which warrants a conditional probability calculation. Consider two sets in $\binom{\mathcal{P}(\mathcal{U})\setminus\emptyset}{k}$ that overlap.
We provide this in the following lemma that computes the probability that one of these sets is a matching given that the other is.

\begin{lemma}\label{lemma:conditionalprob}
Let $(\mathcal{U},\mathcal{S})\sim\mathcal{H}(n,M)$ and Let $\mathcal{S}_k,\mathcal{T}_k$ be arbitrary sets  from $\binom{\mathcal{P}(\mathcal{U})\setminus\emptyset}{k}$, let $\ell=|\mathcal{S}_k\cap\mathcal{T}_k|$, and let $t=|\bigcup_{T\in \mathcal{S}_k\cap\mathcal{T}_k}T|$.
Then $$\Pr[\mathsf{Match}(\mathcal{S}_k)\mid \mathsf{Match}(\mathcal{T}_k)]\le\frac{\sum_{j=1}^{n-t}\binom{n-t}{j}\left\{ {j \atop k-\ell} \right\}}{\binom{2^n-k-1}{k-\ell}}.$$
\end{lemma}
\begin{proof}
Consider the probability that $\mathcal{S}_k$ is a subset of $\mathcal{S}$ given that $\mathcal{T}_k$ is a matching (and hence a subset of $\mathcal{S}$).
Since $\mathcal{T}_k\subseteq \mathcal{S}$, we consider the number of ways to select $\mathcal{S}_k\setminus\mathcal{T}_k$ from  $\mathcal{P}(\mathcal{U})\setminus\mathcal{T}_k\setminus\emptyset$, which is given by $\binom{2^n-k-1}{k-\ell}$.
Of these, there are then $\binom{M-k}{k-\ell}$ ways for $\mathcal{S}_k\setminus\mathcal{T}_k$ to be a subset of $\mathcal{S}\setminus\mathcal{T}_k$, and since all $(\mathcal{U},\mathcal{S})$ are equiprobable, the probability that $\mathcal{S}_k\setminus\mathcal{T}_k$ is a subset of $\mathcal{S}\setminus\mathcal{T}_k$ is given by $\frac{\binom{M-k}{k-\ell}}{\binom{2^n-k-1}{k-\ell}}$.

Now assume $\mathcal{S}_k\subseteq \mathcal{S}$.
Since $\mathcal{T}_k$ is a fixed matching, there are only $\binom{M-k}{k-\ell}$ ways to choose the remaining sets of $\mathcal{S}_k\setminus\mathcal{T}_k$.
Moreover, if $\mathcal{S}_k$ is a matching, then all hyperedges in $\mathcal{S}_k$ must be pairwise disjoint, and so the edges of $\mathcal{S}_k\setminus\mathcal{T}_k$ may not be incident to any vertices from $\bigcup_{T\in\mathcal{S}_k\cap\mathcal{T}_k}T$.
As such, we must exclude these $t$ vertices from the vertex set when counting the number of partitions of  size $k-\ell$.
Let $n_{\mathcal{S}'}=|\bigcup_{S_i\in\mathcal{S}\setminus(\mathcal{S}_k\cap\mathcal{T}_k)}S_i|$.
Then it follows that $n_{\mathcal{S}'}\le n-t$.
Again since all $(\mathcal{U},\mathcal{S})$ are equiprobable, we compute our desired probability as
$$\Pr[\mathsf{Match}(\mathcal{S}_k)|\mathsf{Match}(\mathcal{T}_k)]$$
$$=\Pr[\mathsf{PWD}(\mathcal{S}_k)\land\mathcal{S}_k\subseteq\mathcal{S}|\mathsf{Match}(\mathcal{T}_k)]$$
$$=\Pr[\mathsf{PWD}(\mathcal{S}_k)|\mathsf{Match}(\mathcal{T}_k)\land\mathcal{S}_k\subseteq\mathcal{S}]\Pr[\mathcal{S}_k\subseteq\mathcal{S}\mid \mathsf{Match(\mathcal{T}_k})]$$
$$=\frac{\text{\# of partitions over all subsets of $\bigcup_{S_i\in\mathcal{S}}S_i$ of size } k \text{ that are a superset of } \mathcal{S}_k\cap\mathcal{T}_k}{\binom{M-k}{k-\ell}}\cdot\frac{\binom{M-k}{k-\ell}}{\binom{2^n-k-1}{k-\ell}}$$
$$=\frac{\text{\# of partitions over all subsets of $\bigcup_{S_i\in\mathcal{S}}S_i$ of size } k-\ell \text{ that are pairwise disjoint with } \mathcal{S}_k\cap\mathcal{T}_k}{\binom{2^n-k-1}{k-\ell}}$$
$$=\frac{\text{\# of partitions over all subsets of $\bigcup_{S_i\in\mathcal{S}\setminus(\mathcal{S}_k\cap\mathcal{T}_k)}S_i$ of size } k-\ell }{\binom{2^n-k-1}{k-\ell}}$$
$$=\frac{\sum_{j=1}^{n_{\mathcal{S}'}}\binom{n_{\mathcal{S}'}}{j}\left\{ {j \atop k-\ell} \right\}}{\binom{2^n-k-1}{k-\ell}}$$
$$\le\frac{\sum_{j=1}^{n-t}\binom{n-t}{j}\left\{ {j \atop k-\ell} \right\}}{\binom{2^n-k-1}{k-\ell}}$$
\end{proof}

We now bound the variance from above with the following lemma.

\begin{lemma}\label{lemma:varbound}
For any $1\le f\le M$, we have
    $$\mathsf{Var}[X_f]\le\mathbb{E}[X_f] + \mathbb{E}[X_f]\sum_{\ell=2}^{f}{f\choose \ell}\sum_{j=1}^{n-\ell}\binom{n-\ell}{j}\left\{ {j \atop f-\ell} \right\}.$$
\end{lemma}
\begin{proof}
Let random variable $Y_{\mathcal{S}_f}=1$ if $\mathcal{S}_f\subseteq\mathcal{P}(\mathcal{U})\setminus\emptyset$ is a matching of size $f$, and 0 otherwise.
Then we have
$$
\mathsf{Var}[X_f] =\mathsf{Var}\left[\sum_{\mathcal{S}_f\subseteq\mathcal{S}}Y_{\mathcal{S}_f}\right]
= \sum_{\mathcal{S}_f\subseteq\mathcal{S}} \mathsf{Var}[Y_{\mathcal{S}_f}] + \sum_{\mathcal{S}_f,\mathcal{T}_f\subseteq\mathcal{S}:\mathcal{S}_f\ne\mathcal{T}_f}\mathsf{Cov}[Y_{\mathcal{S}_f}, Y_{\mathcal{T}_f}]
$$

Let $Y_{\mathcal{S}_f}\sim Y_{\mathcal{T}_f}$ denote that $Y_{\mathcal{S}_f}$ and $Y_{\mathcal{T}_f}$ are \emph{not} independent, that is, the sets are distinct and have a non-empty intersection of at least two elements.
Note that $\mathsf{Cov}[Y_{\mathcal{S}_f}, Y_{\mathcal{T}_f}]=0$ unless $Y_{\mathcal{S}_f}\sim Y_{\mathcal{T}_f}$.
For given $\mathcal{S}_k,\mathcal{T}_k$, let $\ell(\mathcal{S}_k,\mathcal{T}_k)=|\mathcal{S}_k\cap\mathcal{T}_k|$ and let $t(\mathcal{S}_k,\mathcal{T}_k)=|\bigcup_{T\in\mathcal{S}_k\cap\mathcal{T}_k}T|$.
Then since each $Y_{\mathcal{S}_f}$ is in $\{0,1\}$ and by Lemma~\ref{lemma:conditionalprob}, we have 
\begin{align*}
    &\sum_{\mathcal{S}_f\subseteq\mathcal{S}} \mathsf{Var}[Y_{\mathcal{S}_f}] + \sum_{\mathcal{S}_f,\mathcal{T}_f\subseteq\mathcal{S}:\mathcal{S}_f\sim\mathcal{T}_f}\mathsf{Cov}[Y_{\mathcal{S}_f}, Y_{\mathcal{T}_f}]\\
   &=\sum_{\mathcal{S}_f\subseteq\mathcal{S}} \mathbb{E}[(Y_{\mathcal{S}_f}-\mathbb{E}[Y_{\mathcal{S}_f}])^2] + \sum_{\mathcal{S}_f,\mathcal{T}_f\subseteq\mathcal{S}:\mathcal{S}_f\sim\mathcal{T}_f}\mathbb{E}[(Y_{\mathcal{S}_f}-\mathbb{E}[Y_{\mathcal{S}_f}])(Y_{\mathcal{T}_f}-\mathbb{E}[Y_{\mathcal{T}_f}])]\\
&\le\sum_{\mathcal{S}_f\subseteq\mathcal{S}} \mathbb{E}[Y_{\mathcal{S}_f}^2] + \sum_{\mathcal{S}_f,\mathcal{T}_f\subseteq\mathcal{S}:\mathcal{S}_f\sim\mathcal{T}_f}\mathbb{E}[Y_{\mathcal{S}_f}Y_{\mathcal{T}_f}] \\
&=\sum_{\mathcal{S}_f\subseteq\mathcal{S}} \mathbb{E}[Y_{\mathcal{S}_f}] + \sum_{\mathcal{S}_f,\mathcal{T}_f\subseteq\mathcal{S}:\mathcal{S}_f\sim\mathcal{T}_f}\mathbb{E}[Y_{\mathcal{S}_f}Y_{\mathcal{T}_f}] \\
&=\mathbb{E}[X_f] + \sum_{\mathcal{S}_f,\mathcal{T}_f\subseteq\mathcal{S}:\mathcal{S}_f\sim\mathcal{T}_f}\Pr[Y_{\mathcal{S}_f}=1\land Y_{\mathcal{T}_f}=1] \\
&=\mathbb{E}[X_f] + \sum_{\mathcal{S}_f,\mathcal{T}_f\subseteq\mathcal{S}:\mathcal{S}_f\sim\mathcal{T}_f}\Pr[Y_{\mathcal{T}_f}=1\mid Y_{\mathcal{S}_f}=1]\Pr[Y_{\mathcal{S}_f}=1] \\
&=\mathbb{E}[X_f] + \sum_{\mathcal{S}_f\subseteq\mathcal{S}}\left(\Pr[Y_{\mathcal{S}_f}=1]\sum_{\mathcal{T}_f\subseteq\mathcal{S}:\mathcal{T}_f\sim\mathcal{S}_f}\Pr[Y_{\mathcal{T}_f}=1\mid Y_{\mathcal{S}_f}=1]\right) \\
&=\mathbb{E}[X_f] + \sum_{\mathcal{S}_f\subseteq\mathcal{S}}\left(\Pr[Y_{\mathcal{S}_f}=1]\sum_{\mathcal{T}_f\subseteq\mathcal{S}:\mathcal{T}_f\sim\mathcal{S}_f}\frac{\sum_{j=1}^{n-t(\mathcal{S}_f,\mathcal{T}_f)}\binom{n-t(\mathcal{S}_f,\mathcal{T}_f)}{j}\left\{ {j \atop f-\ell(\mathcal{S}_f,\mathcal{T}_f)} \right\}}{\binom{2^n-f-1}{f-\ell(\mathcal{S}_f,\mathcal{T}_f)}}\right) \\
\end{align*}
We now count the number of ways to choose an overlapping set $\mathcal{T}_f$ given a fixed $\mathcal{S}_f$.
If $2\le\ell< f$ is the size of the intersection, then there are $\binom{f}{\ell}$ ways for the sets to overlap and $f-1$ values of $\ell$ to consider.
We first count the number of ways to choose the elements of $\mathcal{T}_f\setminus\mathcal{S}_f$.
Note that $\mathcal{T}_f\setminus\mathcal{S}_f$ cannot include any set in $\bigcup_{T\in\mathcal{S}_f\cap\mathcal{T}_f}T$, hence the sets in $\mathcal{T}_f\setminus\mathcal{S}_f$ may be selected from $2^n-f-1$ elements.
As there are $f-\ell$ such sets, the number of ways to select these sets is given by $\binom{2^n-f-1}{f-\ell}$.
Moreover, we note that $(*)$ $t(\mathcal{S}_f,\mathcal{T}_f)\geq\ell$ since each set in the intersection must have at least one distinct element since $\mathcal{S}_f$ is a hyper-matching.
That is, since the sets in $\mathcal{S}_f\cap\mathcal{T}_f$ are unique pairwise disjoint elements of $\mathcal{P}(\mathcal{U})\setminus\emptyset$, set $\mathcal{S}_f\cap\mathcal{T}_f$ has a system of distinct representatives of size $\ell$.
As such, we have
\begin{align*}
    &\mathsf{Var}[X_f]\le\\
    &\mathbb{E}[X_f] + \sum_{\mathcal{S}_f\subseteq\mathcal{S}}\left(\Pr[Y_{\mathcal{S}_f}=1]\sum_{\mathcal{T}_f\subseteq\mathcal{S}:\mathcal{T}_f\sim\mathcal{S}_f}\frac{\sum_{j=1}^{n-t(\mathcal{S}_f,\mathcal{T}_f)}\binom{n-t(\mathcal{S}_f,\mathcal{T}_f)}{j}\left\{ {j \atop f-\ell(\mathcal{S}_f,\mathcal{T}_f)} \right\}}{\binom{2^n-f-1}{f-\ell(\mathcal{S}_f,\mathcal{T}_f)}}\right) \\
    &=\mathbb{E}[X_f] + \sum_{\mathcal{S}_f\subseteq\mathcal{S}}\left(\Pr[Y_{\mathcal{S}_f}=1]\sum_{\ell=2}^{f-1}{f\choose \ell}{2^n-f-1 \choose f - \ell}\frac{\sum_{j=1}^{n-t(\mathcal{S}_f,\mathcal{T}_f)}\binom{n-t(\mathcal{S}_f,\mathcal{T}_f)}{j}\left\{ {j \atop f-\ell} \right\}}{\binom{2^n-f-1}{f-\ell}}\right) \\
    &<\mathbb{E}[X_f] + \sum_{\mathcal{S}_f\subseteq\mathcal{S}}\left(\Pr[Y_{\mathcal{S}_f}=1]\sum_{\ell=2}^{f}{f\choose \ell}{2^n-f-1 \choose f - \ell}\frac{\sum_{j=1}^{n-t(\mathcal{S}_f,\mathcal{T}_f)}\binom{n-t(\mathcal{S}_f,\mathcal{T}_f)}{j}\left\{ {j \atop f-\ell} \right\}}{\binom{2^n-f-1}{f-\ell}}\right) \\
    &\le_{(*)}\mathbb{E}[X_f] + \sum_{\mathcal{S}_f\subseteq\mathcal{S}}\left(\Pr[Y_{\mathcal{S}_f}=1]\sum_{\ell=2}^{f}{f\choose \ell}{2^n-f-1 \choose f - \ell}\frac{\sum_{j=1}^{n-\ell}\binom{n-\ell}{j}\left\{ {j \atop f-\ell} \right\}}{\binom{2^n-f-1}{f-\ell}}\right) \\
    &=\mathbb{E}[X_f] + \sum_{\ell=2}^{f}{f\choose \ell}\sum_{j=1}^{n-\ell}\binom{n-\ell}{j}\left\{ {j \atop f-\ell} \right\}\sum_{\mathcal{S}_f\subseteq\mathcal{S}}\left(\Pr[Y_{\mathcal{S}_f}=1]\right) \\
    &=\mathbb{E}[X_f] + \mathbb{E}[X_f]\sum_{\ell=2}^{f}{f\choose \ell}\sum_{j=1}^{n-\ell}\binom{n-\ell}{j}\left\{ {j \atop f-\ell} \right\} \\
\end{align*}
\end{proof}

With $\mathsf{Var}[X_f]$ bounded from above, by Fact~\ref{fact:chebyshev} we need only find the term $f$ such that
    $$ \frac{\mathsf{Var}[X_f]}{(\mathbb{E}[X_f])^2}\le\frac{\mathbb{E}[X_f]+\mathbb{E}[X_f]\sum_{\ell=2}^{f}{f\choose \ell}\sum_{j=1}^{n-\ell}\binom{n-\ell}{j}\left\{ {j \atop f-\ell} \right\}}{(\mathbb{E}[X_f])^2}$$$$=\frac{1}{\mathbb{E}[X_f]}+\frac{\sum_{\ell=2}^{f}{f\choose \ell}\sum_{j=1}^{n-\ell}\binom{n-\ell}{j}\left\{ {j \atop f-\ell} \right\}}{\mathbb{E}[X_f]}$$
tends to $0$ in the limit.
We handle each of these two terms separately.
\begin{lemma}\label{lemma:iverseexpect}
If $M=\Theta(2^n)$, then the equation $\lim_{n\to\infty}1/\mathbb{E}[X_f]=0$ is satisfied if $f=n^{\frac12-\delta}$ for some $\delta>0$.
\end{lemma}
\begin{proof}
Suppose $M=c2^n$ for some constant fraction $c$.
Then by Lemma~\ref{lemma:expectation} and Facts~\ref{fact:stirlingbound},\ref{fact:binomial}, we have
\begin{align*}
    &\frac{1}{\mathbb{E}[X_f]}\\
    &\le \frac{1}{\frac{M}{2^n}{M \choose f}\frac{\sum_{j=1}^n{n \choose j}\left\{ {j \atop f} \right\}}{\binom{2^n-1}{f}}}\\
    &= \frac{1}{c{M \choose f}\frac{\sum_{j=1}^n{n \choose j}\left\{ {j \atop f} \right\}}{\binom{2^n-1}{f}}}\\
    &= \frac{\binom{2^n-1}{f}}{c{M \choose f}\sum_{j=1}^n{n \choose j}\left\{ {j \atop f} \right\}}\\
    &= \frac{\binom{2^n-1}{f}}{c{c2^n \choose f}\sum_{j=1}^n{n \choose j}\left\{ {j \atop f} \right\}}\\
    &\le \frac{\frac{(2^n-1)^fe^f}{f^f}}{c\frac{c^f2^{nf}}{f^f}\sum_{j=1}^n{n \choose j}\left\{ {j \atop f} \right\}}\\
    &= \frac{(2^n-1)^fe^f}{c^{f+1}2^{nf}\sum_{j=1}^n{n \choose j}\left\{ {j \atop f} \right\}}\\
    &\le \frac{e^f}{c^{f+1}\sum_{j=1}^n{n \choose j}\left\{ {j \atop f} \right\}}\\
    &\le \frac{e^f}{c^{f+1}\sum_{j=1}^n{n \choose j}\left(\frac12(f^2+f+2)f^{j-f-1}-1\right)}\\
    &\le \frac{e^f}{c^{f+1}\sum_{j=1}^n{n \choose j}\left(\frac12(f^2+f+2)f^{-f}-1\right)}\\
    &= \frac{e^f}{c^{f+1}\left(\frac12(f^2+f+2)f^{-f}-1\right)\sum_{j=1}^n{n \choose j}}\\
    &\le \frac{e^f}{c^{f+1}\left(\frac12f^{-f}-1\right)\sum_{j=1}^n{n \choose j}}\\
    &= \frac{(\frac12f^{-f}-1)^{-1}c^{-f-1}e^f}{\sum_{j=1}^n{n \choose j}}\\
    &= \frac{(\frac12f^{-f}-1)^{-1}c^{-f-1}e^f}{2^n-1}\\
    &= \frac{\Theta(\exp(f\ln{f}+f-(f+1)\ln{c})}{\Theta(\exp(n\ln2))}\\
\end{align*}
The expression then tends to zero for $f=o(\sqrt{n})$ and hence $f=n^{\frac12-\delta}$ for $\delta>0$ suffices.
\end{proof}

We now address the second term of the variance bound.
\begin{lemma}\label{lemma:varratio}
 If $M=\Theta(2^n)$, the equation $\lim_{n\to\infty}\frac{\sum_{\ell=2}^{f}{f\choose \ell}\sum_{j=1}^{n-\ell}\binom{n-\ell}{j}\left\{ {j \atop f-\ell} \right\}}{\mathbb{E}[X_f]}=0$ is satisfied if $f=n^{\frac12-\delta}$ for some $\delta>0$.
\end{lemma}
\begin{proof}
Let $M=c2^{n}$ for a constant fraction $c$.
Then by Lemma~\ref{lemma:expectation} and Facts~\ref{fact:stirlingbound},\ref{fact:binomial}, we have
\begin{align*}
    &\frac{\sum_{\ell=2}^{f}{f\choose \ell}\sum_{j=1}^{n-\ell}\binom{n-\ell}{j}\left\{ {j \atop f-\ell} \right\}}{\mathbb{E}[X_f]}
    \le\frac{\sum_{\ell=2}^{f}{f\choose \ell}\sum_{j=1}^{n-\ell}\binom{n-\ell}{j}\left\{ {j \atop f-\ell} \right\}}{\frac{M}{2^n}{M \choose f}\frac{\sum_{j=1}^n{n \choose j}\left\{ {j \atop f} \right\}}{\binom{2^n-1}{f}}}
    =\frac{\sum_{\ell=2}^{f}{f\choose \ell}\sum_{j=1}^{n-\ell}\binom{n-\ell}{j}\left\{ {j \atop f-\ell} \right\}}{c{M \choose f}\frac{\sum_{j=1}^n{n \choose j}\left\{ {j \atop f} \right\}}{\binom{2^n-1}{f}}}
    \\
    &\le \frac{\sum_{\ell=2}^{f}{f\choose \ell}\sum_{j=1}^{n-\ell}\binom{n-\ell}{j}\frac12\binom{j}{f-\ell}(f-\ell)^{j-f+\ell}}{c{M \choose f}\frac{\sum_{j=1}^n{n \choose j}\left(\frac12(f^2+f+2)f^{j-f-1}-1\right)}{\binom{2^n-1}{f}}}
    \le \frac{\sum_{\ell=2}^{f}{f\choose \ell}\sum_{j=1}^{n-\ell}\binom{n-\ell}{j}\frac12\left(\frac{je}{f-\ell}\right)^{f-\ell}(f-\ell)^{j-f+\ell}}{c{M \choose f}\frac{\sum_{j=1}^n{n \choose j}\left(\frac12(f^2+f+2)f^{j-f-1}-1\right)}{\binom{2^n-1}{f}}}\\
    &= \frac{\sum_{\ell=2}^{f}{f\choose \ell}\sum_{j=1}^{n-\ell}\binom{n-\ell}{j}\frac12(je)^{f-\ell}(f-\ell)^{j-2f+2\ell}}{c{M \choose f}\frac{\sum_{j=1}^n{n \choose j}\left(\frac12(f^2+f+2)f^{j-f-1}-1\right)}{\binom{2^n-1}{f}}}
    \le \frac{\sum_{\ell=2}^{f}{f\choose \ell}\sum_{j=1}^{n-\ell}\binom{n-\ell}{j}\frac12(ne)^{f-\ell}(f-\ell)^{j-2f+2\ell}}{c{M \choose f}\frac{\sum_{j=1}^n{n \choose j}\left(\frac12(f^2+f+2)f^{j-f-1}-1\right)}{\binom{2^n-1}{f}}}\\
    &= \frac{\sum_{\ell=2}^{f}{f\choose \ell}\frac12(ne)^{f-\ell}\sum_{j=1}^{n-\ell}\binom{n-\ell}{j}(f-\ell)^{j-2f+2\ell}}{c{M \choose f}\frac{\sum_{j=1}^n{n \choose j}\left(\frac12(f^2+f+2)f^{j-f-1}-1\right)}{\binom{2^n-1}{f}}}
    = \frac{\sum_{\ell=2}^{f}{f\choose \ell}\frac12(ne)^{f-\ell}(f-\ell)^{2\ell-2f}\sum_{j=1}^{n-\ell}\binom{n-\ell}{j}(f-\ell)^{j}}{c{M \choose f}\frac{\sum_{j=1}^n{n \choose j}\left(\frac12(f^2+f+2)f^{j-f-1}-1\right)}{\binom{2^n-1}{f}}}\\
    &\le_{(Fact~\ref{fact:binomial})} \frac12(ne)^{f}\frac{\sum_{\ell=2}^{f}{f\choose \ell}(f-\ell)^{2\ell-2f}(f-\ell+1)^{n-\ell}}{c{M \choose f}\frac{\sum_{j=1}^n{n \choose j}\left(\frac12(f^2+f+2)f^{j-f-1}-1\right)}{\binom{2^n-1}{f}}}
    \le_{2\le\ell \le f-1} \frac12(ne)^{f}\frac{\sum_{\ell=2}^{f}{f\choose \ell}(f-\ell)^{-2}f^{n-2}}{c{M \choose f}\frac{\sum_{j=1}^n{n \choose j}\left(\frac12(f^2+f+2)f^{j-f-1}-1\right)}{\binom{2^n-1}{f}}}\\
    &\le \frac12(ne)^{f}\frac{\sum_{\ell=2}^{f}{f\choose \ell}f^{n-2}}{c{M \choose f}\frac{\sum_{j=1}^n{n \choose j}\left(\frac12(f^2+f+2)f^{j-f-1}-1\right)}{\binom{2^n-1}{f}}}
= \frac12(ne)^{f}f^{n-2}\frac{\sum_{\ell=2}^{f}{f\choose \ell}}{c{M \choose f}\frac{\sum_{j=1}^n{n \choose j}\left(\frac12(f^2+f+2)f^{j-f-1}-1\right)}{\binom{2^n-1}{f}}}\\
    &\le \frac12(ne)^{f}f^{n-2}\frac{2^f}{c{M \choose f}\frac{\sum_{j=1}^n{n \choose j}\left(\frac12(f^2+f+2)f^{j-f-1}-1\right)}{\binom{2^n-1}{f}}}
    = \frac12(ne)^{f}f^{n-2}\frac{2^f}{c{M \choose f}\frac{\sum_{j=1}^n\left(\frac12{n \choose j}(f^2+f+2)f^{j-f-1}-{n \choose j}\right)}{\binom{2^n-1}{f}}}\\
    &= \frac12(ne)^{f}f^{n-2}\frac{2^f}{c{M \choose f}\frac{\sum_{j=1}^n\frac12{n \choose j}(f^2+f+2)f^{j-f-1}-\sum_{j=1}^n{n \choose j}}{\binom{2^n-1}{f}}}
    = \frac12(ne)^{f}f^{n-2}\frac{2^f}{c{M \choose f}\frac{\frac{(f^2+f+2)}{2f}\sum_{j=1}^n{n \choose j}f^{j-f}-\sum_{j=1}^n{n \choose j}}{\binom{2^n-1}{f}}}\\
    &= \frac12(ne)^{f}f^{n-2}\frac{2^f}{c{M \choose f}\frac{\frac{(f^2+f+2)}{2f}\sum_{j=1}^n{n \choose j}f^{j-f}-2^n+1}{\binom{2^n-1}{f}}}
    \le \frac12(ne)^{f}f^{n-2}\frac{2^f}{c{M \choose f}\frac{\frac{(f^2+f+2)}{2f}\sum_{j=1}^n{n \choose j}f^{j}-2^n+1}{\binom{2^n-1}{f}}}\\
    &= \frac12(ne)^{f}f^{n-2}\frac{2^f}{c{M \choose f}\frac{\frac{(f^2+f+2)}{2f}(f+1)^n-2^n}{\binom{2^n-1}{f}}}
    \le \frac12(ne)^{f}f^{n-2}\frac{2^f}{c{M \choose f}\frac{(f+1)^n-2^n}{\binom{2^n-1}{f}}}\\
    &= \frac12(ne)^{f}f^{n-2}\frac{2^f}{c{c2^n \choose f}\frac{f^n-2^n}{\binom{2^n-1}{f}}}
    \le \frac12(ne)^{f}f^{n-2}\frac{2^f}{\frac{c^{f+1}2^{nf}}{f^f}\frac{(f+1)^n-2^n}{\frac{(2^n-1)^fe^f}{f^f}}}
    = \frac12(ne)^{f}f^{n-2}\frac{2^f(2^n-1)^fe^f}{c^{f+1}2^{nf}((f+1)^n-2^n)}\\
    &\le \frac12(ne)^{f}f^{n-2}\frac{2^fe^f}{c^{f+1}((f+1)^n-2^n)}
    \le\frac{n^{f+2f/\ln{n}+n\ln{f}/\ln{n}+f\ln{2}/\ln{n}-(f+1)\ln{c}/\ln{n}}}{(f+1)^n-2^n}\\
\end{align*}
Let $f=n^{\frac12-\delta}$ for $\delta>0$.
Note that $(f+1)^{n}$ asymptotically dominates $2^n$.
Thus, it suffices to consider the following expression. 
\begin{align*}
&\frac{n^{f+2f/\ln{n}+n\ln{f}/\ln{n}+f\ln{2}/\ln{n}-(f+1)\ln{c}/\ln{n}}}{(f+1)^n}\\
&=n^{\frac{f\ln{n}+(2+\ln{2}-\ln{c})f-\ln{c}}{\ln{n}}+n\frac{\ln{\left(\frac{f}{f+1}\right)}}{\ln{n}}}\\
\end{align*}

We now show the following inequality
$$\ln(f)-\ln(f+1)\le -\frac{\ln{f}}{f^{1+\epsilon}},f>0,\text{where }\epsilon>0 \text{ is an arbitrary constant.}$$
It suffices to show
\begin{align*}
    &\ln(f)-\ln(f+1)\le -\frac{\ln{f}}{f^{1+\epsilon}}\\
    &\iff \ln(f+1)-\ln(f)\geq\frac{\ln{f}}{f^{1+\epsilon}}\\
    &\iff \frac{e^{\ln{(f+1)}}}{e^{\ln{f}}}\geq e^{\frac{\ln{f}}{f^{1+\epsilon}}}\\
    &\iff \frac{f+1}{f}\geq e^{\frac{\ln{f}}{f^{1+\epsilon}}}\\
    &\iff 1\geq f\left(e^{\frac{\ln{f}}{f^{1+\epsilon}}}-1\right)\\
\end{align*}

It suffices to show that the function tends to zero. By Fact~\ref{fact:exptaylor},

\begin{align*}
    &\lim_{f\to\infty}\left(f\left(e^{\frac{\ln{f}}{f^{1+\epsilon}}}-1\right)\right)\\
    &=\lim_{f\to\infty}\left(f\left(\left(\sum_{j=0}^\infty\frac{\ln^j{f}}{f^{j+j\epsilon}j!}\right)-1\right)\right)\\
    &=\lim_{f\to\infty}\left(f\left(\sum_{j=1}^\infty\frac{\ln^j{f}}{f^{j+j\epsilon}j!}\right)\right)\\
    &=\lim_{f\to\infty}\left(\sum_{j=1}^\infty\frac{\ln^j{f}}{f^{j+j\epsilon-1}j!}\right)\\
    &=_{(*)}\sum_{j=1}^\infty\left(\lim_{f\to\infty}\frac{\ln^j{f}}{f^{j+j\epsilon-1}j!}\right)\\
    &=_{(**)}0
\end{align*}
Where $(*)$ holds since the Taylor series converges absolutely, and $(**)$ holds as the denominator of each term is raised to a power that is strictly greater than zero since $j$ counts from 1 and $\epsilon>0$.
Furthermore, for large $f$ the numerator of each term grows strictly slower than its denominator. 
Hence for $f \geq f_0$, where $f_0$ is sufficiently large, the function gets arbitrarily close to $0$ and so $f\left(e^{\frac{\ln f}{f^{1+\epsilon}}}-1\right)\le 1$ in this range.

Given our choice of $f$, we have
\begin{align*}
&n^{\frac{n^{\frac{1}{2}-\delta}\ln{n}+(2+\ln{2}-\ln{c})n^{\frac{1}{2}-\delta}-\ln{c}}{\ln{n}}+n\frac{\ln{\left(\frac{n^{\frac{1}{2}-\delta}}{n^{\frac{1}{2}-\delta}+1}\right)}}{\ln{n}}}\\
&\le n^{\frac{n^{\frac{1}{2}-\delta}\ln{n}+(2+\ln{2}-\ln{c})n^{\frac{1}{2}-\delta}-\ln{c}}{\ln{n}}-\frac{n}{n^{(\frac{1}{2}-\delta)(1+\epsilon)}\ln{n}}(\ln{n^{\frac{1}{2}-\delta}})}\\
&= n^{\frac{n^{\frac{1}{2}-\delta}\ln{n}+(2+\ln{2}-\ln{c})n^{\frac{1}{2}-\delta}-\ln{c}}{\ln{n}}-\frac{n}{n^{(\frac12+\frac{\epsilon}{2}-\delta-\delta\epsilon)}\ln{n}}(\ln{n^{\frac{1}{2}-\delta}})}\\
&= n^{\Theta(n^{\frac{1}{2}-\delta})-\Theta(n^{\frac12-\frac{\epsilon}{2}+\delta+\delta\epsilon})}\\
&=_{(*)} n^{\Theta(n^{\frac{1}{2}-\delta})-\Theta(n^{\frac12+\frac{\delta}{2}+\delta^2})}\\
\end{align*}
Where $(*)$ holds since we may choose $\epsilon=\delta>0$.
Now since $\Theta(n^{\frac{1}{2}-\delta}),\Theta(n^{\frac12+\frac{\delta}{2}+\delta^2})$ are increasing functions that tend to infinity, the difference is asymptotically negative, and hence the expression tends to zero in the limit.

\end{proof}

Putting it all together, we now present the proof of Theorem~\ref{theroem:main}.
\begin{proof}
    We consider both cases separately.\\
    Case 1: $M=o(1.155^n)$. 
    Since a single element of the powerset of $\mathcal{U}$ is trivially a matching, the hyper-matching number is at least 1.
    By Lemma~\ref{lemma:smallM}, the hyper-matching number is at most 1 almost surely.
    Hence, the hyper-matching number is 1 almost surely.\\
    Case 2: $M=\Theta(2^{n})$. By Lemmas~\ref{lemma:varbound},\ref{lemma:iverseexpect},\ref{lemma:varratio}, Fact~\ref{fact:chebyshev}, and Observation~\ref{obs:upperbound}, we have the hyper-matching number is almost surely in $[\Omega(n^{\frac12-\delta}),n]$ for arbitrary $\delta>0$.
    Moreover, by Observation~\ref{obs:upperbound}, the bound from above is tight.
\end{proof}

\section{Empirical Studies}\label{sec:experimentalvalidation}
In this section we first present an empirical study that validates the surprising results from Lemma~\ref{lemma:smallM}. Namely that for small $M$ the hyper-matching number converges to unity in the limit.
Secondly, we present a study that analyzes the behavioral gap empirically.
Note that as the maximum cardinality matching problem for hypergraphs is \textsf{NP}-hard, solving large instances is prohibitively expensive in practice.
Hence, for both studies, we scaled the instance size as large as possible for the computations to terminate in a reasonable amount of time.
In order to compute the hyper-matching number, we formulated the following Integer Linear Program (ILP).
\begin{align*}
    \text{max} & \sum_{S \in \mathcal{S}}x_{S} & \\
    \text{s.t.} & \sum_{S: s \in S} x_{S} \leq 1 \quad \forall s \in \mathcal{U} \\
    &x_{S} \in \{0, 1\}
\end{align*}
The ILP was then solved with Pyomo (see \cite{bynum2021pyomo} and \cite{hart2011pyomo}) to determine the maximum cardinality matching for an instance. Because the number of edges can be exponential in size, this ILP is also potentially exponential in size and difficult to even store in memory, much less compute the optimal value for large (or even modest) values of $n$.

\subsection{Validation of Convergence to Unity for Small $M$}\label{sec:smallM}
To validate the surprising result from Lemma~\ref{lemma:smallM}, we varied $n$ from 5 to 55, and for each increment (increasing $n$ by one), we generated uniformly 30 different samples of random hypergraphs with exactly $\lfloor 1.154^n\rfloor =o(1.155^n)$ hyperedges.
We present the results in Figure~\ref{fig:varyn}, where we plot the sample mean with 95\% confidence intervals.
We see that for all $n$ considered the average hyper-matching number remains well below 2.
Although this does not confirm the behavior in the limit, it does confirm that the hyper-matching number remains close to 1 when $M=o(1.155^n)$ for hypergraphs as large as 55 vertices, and suggests that the trend may continue.

\begin{figure}
    \centering
    \includegraphics[width=\linewidth]{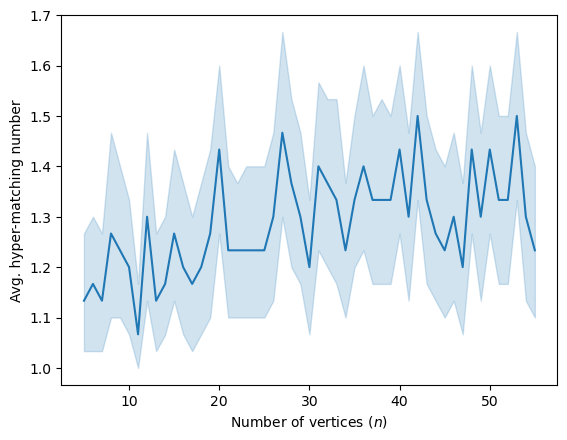}
    \caption{A plot of the average (taken over 30 separate trials per value of $n$) and 95\% confidence intervals for the exact hyper-matching number given that the number of hyperedges is exactly $\lfloor 1.154^n\rfloor$. The plot shows that the average hyper-matching number remains close to 1 for $n\le 55$.}
    \label{fig:varyn}
\end{figure}

\subsection{Empirical Study of the Behavioral Gap}\label{sec:behavioralexp}
In this section we perform an empirical study of the behavioral gap.
Here we fix $n=13$ and then vary $M$ up to $2^n-1$ in increments of ten.
For each sample point of $M$ we generate thirty random instances and plot the sample mean with 95\% confidence intervals.
We present the results in Figures~\ref{figure:Mtrend} and~\ref{figure:localtrend}. 
We sampled $M$ elements from $\mathcal{P}(\mathcal{U}) \setminus \emptyset$ over 30 different trials to collect various hyper-matching numbers parameterized on different values of $M$. 



\begin{figure*}[t!]
    \centering
    \begin{subfigure}[t]{0.5\textwidth}
        \centering
        \includegraphics[width=8cm]{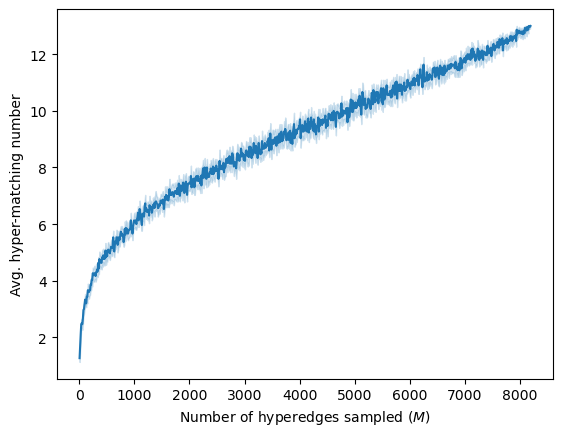}
        \caption{}
        \label{figure:Mtrend}
    \end{subfigure}%
    ~ 
    \begin{subfigure}[t]{0.5\textwidth}
        \centering
        \includegraphics[width=8cm]{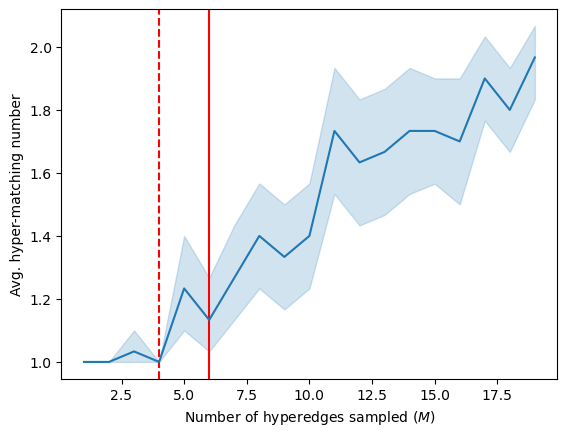}
        \caption{}
        \label{figure:localtrend}
    \end{subfigure}
    \caption{(a) A plot of the average (taken over 30 separate trials per value of $M$) and 95\% confidence intervals for the exact hyper-matching number given $n=13$ and $M$ ranging between $1.155^{13}$ and $2^{13}-1$ by increments of $10$. (b) A close up plot of the average (taken over 30 separate trials per value of $M$) and 95\% confidence intervals for the exact hyper-matching number given $n=13$ and $M$ ranging between $1$ and $20$. Note that $M=1.155^{13}\approx6$ marks approximately the lower bound of the behavioral gap (shown in red) as predicted by our analysis. Empirically, we observe the abrupt increase at $M=4$ (shown by the dashed red line).}
\end{figure*}

As it can be seen in Figure~\ref{figure:Mtrend}, the growth sharply increases at $M\approx 4$, which is slightly below the predicted $M=6 \approx 1.155^{13}$. There is a segment of sublinear growth roughly between $M \approx 6$ and $M \approx 2000$, followed by linear growth from $M\approx2000$ and onward. 
Moreover, as it can be seen in Figure~\ref{figure:localtrend}, the hyper-matching number is very close to 1 when $M<4$, which is near where our theory predicts.
These results validate the predictions of Theorem~\ref{theroem:main} outside the behavioral gap; namely that the hyper-matching number grows at least sublinearly beyond the behavioral gap, and remains close to 1 below this region.
Although this plot does not confirm the behavior of the hyper-matching number in the limit, Theorem~\ref{theroem:main} approximately reflects its behavior outside the behavioral gap for $n=13$.
This suggests that the behavior of the hyper-matching number within the gap may also be characterized approximately by the observed curve in Figures~\ref{figure:Mtrend} and~\ref{figure:localtrend}.
The behavioral gap is not yet well understood and is likely an interesting line of future research.

\section{Discussion}\label{sec:discussion}
We have examined critical thresholds for matchings in random hypergraphs based on edge density, and we have established that in cases where the number of edges is $o(1.155^n)$, maximum cardinality hyper-matchings are still (surprisingly) small even with a large number vertices. Furthermore, we have shown that a large increase occurs after the threshold of $o(1.155^n)$.
When $M=\Theta(2^n)$, we found that there almost surely exists a hyper-matching of size $\Omega(n^{\frac12-\delta})$ for arbitrary $\delta>0$.

As mentioned in Section \ref{sec:behavioralexp}, there is much to learn regarding the dynamics of random maximum hyper-matchings within the range of $\Omega(1.155^n)$ to $o(2^n)$, although we have some understanding of how the hyper-matching number behaves in this range. 
We attribute our inability to analyze the range theoretically to the looseness of the Chebyshev inequality.
Unfortunately, stronger concentration inequalities such as Chernoff bounds typically require independence assumptions, which likely do not hold for this problem.

There is also a closely related notion of set packings which has not been fully explored. For the set packing problem~\cite{Karp1972}, we are given a collection of sets with elements belonging to a universe.
We select the largest cardinality subset of these sets such that the pairwise intersections of all selected sets are empty. 
This problem differs from maximum cardinality matching in hypergraphs in that the edge set is a multiset, and therefore can have cardinality larger than $2^n-1$ (due to duplicates).
The critical thresholds for the set packing problem would also be of interest. 

\section*{Acknowledgments}
This work is partially supported by awards NSF CCF-2312537, NSF CCF-2312537, and U.S. ARO MURI W911NF-19-1-0233. Further partial support has been granted by Arizona State University.
The authors would like to acknowledge Professor Andr\'{e}a W. Richa for her valuable feedback.
\bibliographystyle{plain}
\bibliography{bibliography.bib}

\end{document}